\documentclass[11pt]{myclass}
\usepackage{euscript}

\usepackage{amssymb}
\usepackage{epsfig}
\usepackage{color}
\usepackage{xspace}

\usepackage{url}

\newcommand{\eps}{\varepsilon}
\renewcommand{\Pr}{\ensuremath{\textsf{Pr}}}

\renewcommand{\b}[1]{\ensuremath{\mathbb{#1}}}

\newcommand{\wid}{\ensuremath{\textsf{wid}}}

\newcommand{\E}{\mathbf{E}}
\renewcommand{\Pr}{\mathbf{Pr}}
\newcommand{\abs}[1]{\left| #1 \right|}
\definecolor{orange}{rgb}{1,0.5,0}

\newcommand{\ED}{\mathrm{ED}}

\renewcommand{\paragraph}[1]{\vspace{2mm}\noindent{\bf #1}}
\newcommand{\XOR}{\textsf{XOR}\xspace}
\newcommand{\EQ}{\textsf{EQ}\xspace}
\newcommand{\OR}{\textsf{OR}\xspace}
\newcommand{\ORb}{\textsf{OR-board}\xspace}
\newcommand{\AND}{\textsf{AND}\xspace}
\newcommand{\MAJ}{\textsf{MAJ}\xspace}
\newcommand{\BITS}{\textsf{BITS}\xspace}
\newcommand{\CONN}{\textsf{CONN}\xspace}
\newcommand{\DISJ}{\textsf{DISJ}\xspace}
\newcommand{\HH}{\textsf{HH}\xspace}

\begin{document}

\title{\Large Lower Bounds for Number-in-Hand Multiparty Communication Complexity, Made Easy\thanks{A preliminary version of the paper was presented at the ACM-SIAM Symposium on Discrete Algorithms, January 2012}  }
\author{Jeff M. Phillips~\thanks{Jeff M. Phillips's work on this paper was/is supported by a subaward to the University of Utah under NSF award 0937060 to the Computing Research Association and NSF CCF-1350888, IIS-1251019, and ACI-1443046.} \\ School of Computing \\ University of Utah \\ jeffp@cs.utah.edu \and Elad Verbin~\thanks{Elad Verbin acknowledges support from the Danish National Research Foundation and the National Science Foundation of China (under the grant 61061130540) for the Sino-Danish Center for the Theory of Interactive
Computation, within which part of this work was performed. } \\ elad.verbin@gmail.com \and Qin Zhang~\thanks{Corresponding author. Most of the work was done while Qin Zhang was a postdoc in MADALGO, Aarhus University.} \\ Indiana University Bloomington  \\ qzhangcs@indiana.edu}
\date{}

\maketitle

%\pagenumbering{arabic}
%\setcounter{page}{1}%Leave this line commented out.

\begin{abstract} 
In this paper we prove lower bounds on randomized multiparty communication complexity, both in the \emph{blackboard model} (where each message is written on a blackboard for all players to see) and (mainly) in the \emph{message-passing model}, where messages are sent player-to-player. We introduce a new technique for proving such bounds, called \emph{symmetrization}, which is natural, intuitive, and often easy to use.

For example, for the problem where each of $k$ players gets a bit-vector of length $n$, and the goal is to compute the coordinate-wise XOR of these vectors, we prove a tight lower bounds of $\Omega(nk)$ in the blackboard model. For the same problem with AND instead of XOR, we prove a lower bounds of roughly $\Omega(nk)$ in the message-passing model (assuming $k \le n/3200$) and $\Omega(n \log k)$ in the blackboard model. We also prove lower bounds for bit-wise majority, for a graph-connectivity problem, and for other problems; the technique seems applicable to a wide range of other problems as well. All of our lower bounds allow randomized communication protocols with two-sided error.

We also use the symmetrization technique to prove several direct-sum-like results for multiparty communication.
\end{abstract}

\section{Introduction}
\label{sec:intro}
%Part 1: introduction, explanation of why the model is important, why
%it's not trivial given known work, what cool stuff it proves, what
%are the unique phenomena in this model, and some warmup bounds
%(without symmetrization, just easy bounds, e.g. Jeff's Omega(k) lower
%bound for n=2). Also a comparison to known work and what this means
%for streaming, sketching and tracking, both in corollaries and in the
%sense of intuitive connection.

%\elad{important: we have a small bug!!! We need to lose a factor of logk in the simulation, since when Alice talks to Bob she needs to tell him what player the message was supposed to be sent to. This does not happen in the coordinator model, so maybe we should just use the coordinator model. Note that this means the coordinator model is not equivalent to the message-passing model, just because the message-passing model is silly that way. For streaming, players only send messages to a pre-known individual, so no logk factor arises}

%\elad{add generalized inner product. why isn't it in there already? important!}

In this work we consider multiparty communication complexity in the
\emph{number-in-hand} model. In this model, there are $k$ players $\{p_1, \ldots, p_k\}$, each with his own $n$-bit input $x_i \in \{0,1\}^n$. The players wish to collaborate in order to compute a joint function of their inputs, $f(x_1, \ldots, x_k)$. To do so, they are allowed to communicate, until one of them figures out the value of $f(x_1, \ldots, x_k)$ and returns it. All players are assumed to have unlimited computational power, so all we care about is the amount of communication used. There are three variants to this model, according to the mode of communication:
\begin{enumerate}
 \item the \emph{blackboard model}, where any message sent by a player is written on a blackboard visible to all players;
 \item the \emph{message-passing model}, where a player $p_i$ sending a message specifies another player $p_j$ that will receive this message;
 \item the \emph{coordinator model}, where there is an additional $(k+1)$-th player called the \emph{coordinator}, who receives no input. Players can only communicate with the coordinator, and not with each other directly.
\end{enumerate}
We will work in all of these, but will mostly concentrate on the message-passing model and the coordinator model. Note that the coordinator model is almost equivalent to the message-passing model, up to a $\log k$ multiplicative factor, since instead of player $i$ sending message $x$ to player $j$, player $i$ can transmit message $(j,x)$ to the coordinator, and the coordinator forwards it to player $j$.

Lower bounds in the three models above are useful for proving lower bounds on the space usage of streaming algorithms, and for other models as well, as we explain in Section \ref{subsec:motivation}. Most previous lower bounds have been proved in the blackboard model, but lower bounds in the message-passing model and the coordinator model can potentially give higher bounds for all the applications. 
%\qin{though we haven't found any ...}

Note that another, entirely different, model for multiparty communication is the \emph{number-on-forehead} model, where each player can see the inputs of all other players but \emph{not} his own input. This model has important applications for circuit complexity (see e.g.~\cite{KN97}). We do not discuss this model in this paper.

We allow all protocols to be randomized, with public coins, i.e.\ all players have unlimited access to a common infinite string of independent random bits. We allow the protocol to return the wrong answer with probability $\eps$ (which should usually be thought of as a small constant); here, the probability is taken over the sample space of public randomness. Note that the public coin model might seem overly powerful, but in this paper we are mainly interested in proving lower bounds rather than upper bounds, so giving the model such strength only makes our results stronger.

For more on communication complexity, see the book of Kushilevitz and Nisan~\cite{KN97}, and the references therein. We give some more definitions in the preliminaries, in Section \ref{sec:pre}.

\subsection{Warm-Up}

We begin by sketching two lower bounds obtained using our symmetrization technique, both of them for the coordinate-wise $k$-\XOR problem. These lower bounds can be proved without using symmetrization, but their proofs that use symmetrization are particularly appealing.

First consider the following problem: Each player gets a bitvector $x_i \in \{0,1\}^n$ and the goal is to compute the coordinate-wise XOR of these vectors. We operate in the \emph{blackboard model}, where messages are posted on a blackboard for all to see.

\begin{theorem} \label{thm:XOR_bb_intro}
The coordinate-wise $k$-\XOR problem requires communication $\Omega(nk)$ in the blackboard model.
\end{theorem}

To see this, first let us specify the \emph{hard distribution}: we prove the lower bound when the input is drawn from this distribution, and by the easy direction of Yao's Minimax Lemma (see e.g.\ \cite{KN97}), it follows that this lower bound applies for the problem as a whole. The hard distribution we choose is just the distribution where the inputs are independently drawn from the uniform distribution.

To prove the lower bound, consider a protocol $P$ for this $k$-player problem, which works on this distribution, communicates $C(P)$ bits in expectation, and suppose for now that it never makes any errors (it will be easy to remove this assumption). We build from $P$ a new protocol $P'$ for a $2$-player problem. In the $2$-player problem, suppose that Alice gets input $x$ and Bob gets input $y$, where $x,y \in \{0,1\}^n$ are independent random bitvectors. Then $P'$ works as follows: Alice and Bob randomly choose two distinct indices $i,j \in \{1,\ldots,k\}$ using the public randomness, and they simulate the protocol $P$, where Alice plays player $i$ and lets $x_i=x$, Bob plays player $j$ and lets $x_j=y$, and they both play all of the rest of the players; the inputs of the rest of the players is chosen from shared randomness. Alice and Bob begin simulating the running of $P$. Every time player $i$ should speak, Alice sends to Bob the message that player $i$ was supposed to write on the board, and vice versa. When any other player $p_r$ ($r \neq i,j$) should speak, both Alice and Bob know his input so they know what he should be writing on the board, thus no communication is actually needed (this is the key point of the symmetrization technique). A key observation is that the inputs of the $k$ players are uniform and independent and thus entirely symmetrical,\footnote{Here and throughout the paper, \emph{symmetric} means that the inputs are drawn from a distribution where renaming the players does not change the distribution. Namely, a distribution $D$ over $X^n$ is called \emph{symmetric} if exchanging any two coordinates in $D$ keeps the distribution the same.} and since the indices $i$ and $j$ were chosen uniformly at random, then the expected communication performed by the protocol $P'$ is $\E [C(P')] = 2C(P)/k$. Furthermore, when the players have finished simulating $P'$, Alice now knows the bit-wise XOR of all the vectors $x_1,\ldots,x_k$, from which she can easily reconstruct the vector $x_j$ (since she already knows all the other vectors $x_r$ for $r \neq j$). It follows that using $2C(P)/k$ expected communication, Alice has managed to reconstruct Bob's entire input; from an easy information-theoretic argument it follows that $2C(P)/k \ge n$, so $C(P) \ge \Omega(nk)$, proving the theorem.

Extending the above argument to cover the case where the protocol $P$ is allowed to return the wrong answer with probability $\eps$ is also easy, simply by showing that if Alice managed to learn Bob's entire bit-vector with probability $1-\eps$, then $(1-\eps) \cdot n$ bits must have been communicated: this also follows easily from information-theoretic arguments.

Note the crucial fact that the hard distribution we chose is symmetric: if the distribution of the inputs to the $k$ players was not symmetric, then the protocol $P$ could try to deduce the indices $i$ and $j$ from some statistical properties of the observed inputs, and act according to that. Then the best upper bound we could get on the communication complexity of $P'$ would be $C(P') \le C(P)$, which is much too weak.

We have just described the version of the symmetrization method for the blackboard model. A similar (slightly more complicated) line of thinking leads to a lower bound of $\Omega(n \log k)$ on the complexity of the coordinate-wise AND problem in the blackboard model; see Section \ref{sec:AND}.

Let us now sketch the symmetrization method as it is used in the \emph{coordinator model}, where players can only send and receive messages to and from the coordinator. We prove a lower bound on the same problem as above, the coordinate-wise $k$-\XOR problem.

\begin{theorem}
The coordinate-wise $k$-\XOR problem requires communication $\Omega(nk)$ in the coordinator model.
\end{theorem}

Note that this theorem actually follows from Theorem \ref{thm:XOR_bb_intro}, since the blackboard model is stronger than the coordinator model. However, the proof we sketch here is useful as a warmup exercise, since it shows an easy example of the symmetrization technique as applied to the coordinator model. In the paper we prove multiple lower bounds using this technique, most of which do not follow from corresponding bounds in the blackboard model.

To prove this theorem, we use the same hard distribution as above: all inputs are uniform and independent. Let $P$ be a protocol in the coordinator model, which computes the coordinate-wise XOR, uses communication $C(P)$ and for now assume that $P$ never makes any errors. As before, we build a new protocol $P'$ for a 2-player problem. In the $2$-player problem, suppose that Alice gets input $x$ and Bob gets input $y$, where $x,y \in \{0,1\}^n$ are independent random bitvectors. Then $P'$ works as follows: Alice and Bob choose a single index $i \in \{1,\ldots,k\}$ from public randomness. Alice and Bob simulate the protocol $P$, where Alice simulates player $i$ and lets $x_i=x$, and Bob plays \emph{all the rest of the players}, including the coordinator, and chooses their inputs uniformly, conditioned on their XOR being equal to $y$.

To simulate the protocol $P$, whenever player $i$ needs to send a message to the coordinator, then Alice sends a message to Bob, and whenever the coordinator needs to send a message to player $i$, Bob sends a message to Alice. Whenever any player $j \neq i$ needs to speak to the coordinator, no communication is needed, since both are played by Bob. Note again that the distribution of the inputs of the $k$ players is independent uniform  (for this it is crucial to remember that $x$ and $y$ were uniform and independent in the first place). Once again, from reasons of symmetry, since the index $i$ was chosen uniformly and the inputs are symmetric, the expected communication performed by the protocol $P'$ is $\E[ C(P')] \le 2C(P)/k$. Furthermore, at the end of the running of $P'$, Alice knows the value of $x \oplus y$ so she can reconstruct the value of $y$. As before, this implies the theorem. The assumption that we never make errors can once again be easily removed.

\subsubsection{Discussion}
We see that the crux of the symmetrization technique in the coordinator model is to consider the $k$-player problem that we wish to lower-bound, to find a symmetric distribution which is hard for it, to give Alice the input of one player (chosen at random) and Bob the input of all other players, and to prove a lower bounds for this two-player problem. If the lower bound for the two player problem is $L$, the lower bound for the $k$-player problem will be $kL$. For the blackboard model, the proofs have the same outline, except in the $2$-player problem Alice gets the input of one randomly chosen player, Bob gets the input of another, and they both get the inputs of all the rest of the players. There is one important thing to note here: This argument only works when the hard distribution is symmetric.

%Both of these points are limitations of the symmetrization technique. We discuss them in more depth in the Conclusions (Section \ref{sec:conclusion}).

\subsection{Motivation, Previous Work and Related Models}
\label{subsec:motivation}

Communication complexity is a widely studied topic. In multiplayer number-in-hand communication complexity, the most studied mode of communication is the blackboard model. The message-passing model was already considered in \cite{DR98}. (This model can also be called the \emph{private-message model}, but note that this name was used in \cite{GH09,GG07} for a different model.)
The coordinator model can be thought of as a server-site setting~\footnote{This terminology is similar as the standard ``client-server", and is used extensively in the literature.}, where
there is one server and $k$ sites. Each site has gathered $n$ bits of
information and the server wants to evaluate a function on the collection
of these $k \cdot n$ bits.  Each site can only communicate with the server,
and a server can communicate with any site. This server-site model has been widely studied in the databases and distributed computing communities. Work includes computing top-$k$ \cite{CW04,MTW05,SS08} and heavy hitters
\cite{ZOWX06,HYLC11}.

Another closely related model is the {\em distributed monitoring model},
in which we also have one server and $k$ sites. The only difference is that now the computation is dynamic. That is, each site receives a stream of elements over time and the server would like to maintain continuously at all times some function $f$ of all the elements in the $k$ sites. Thus the server-site model can be seen as a one-shot version of the distributed streaming setting. It follows that any communication complexity lower bound in the message-passing model or the coordinator model also hold in the distributed monitoring model. A lot of work on distributed monitoring has been done recently in the theory community and the database community, including maintaining random samplings
\cite{CMYZ10}, frequency moments \cite{CG05,CMY08}, heavy hitters \cite{BO03,KCR06,MSDO05,YZ09,HYZ12},
quantiles \cite{CGMR05,YZ09}, entropy \cite{ABC09}, various sketches \cite{CMZ06,CGMR05} and some
non-linear functions \cite{SSK06,SSK08}.

We will come back to the latter two models in Section \ref{sec:application}. It is interesting to note that despite the large number of upper bounds (i.e.\ algorithms, communication protocols) in the above models,
very few lower bounds have been proved in any of those models,
likely because there were few known techniques to prove such results.

A further application of the message-passing model could be for Secure Multiparty Computation: in this model, there are several players who do not trust each other, but want to compute a joint function of their inputs, with
each of them learning nothing about the inputs of the others players except
what can be learned from the value of the joint function. Obviously, any lower bound in the message passing model immediately implies a lower bound on the amount of communication required for Secure Multiparty Computation. For more on this model, see e.g. \cite{MPC_survey}.

One final application is for the streaming model~\cite{AMS99}. In this model, there is a long stream of data that can only be scanned from left to right. The goal is to compute some function of the stream, and minimize the space usage. It is easy to see that if we partition the stream into $k$ parts and give each part to a different player, then a lower bound of $L$ on the communication complexity of the problem in the coordinator model implies a lower bound of
$L/k$ on the space usage. When $t$ passes over the model are allowed, a lower bound of $L$ in the coordinator model translates to a lower bound of $L/tk$ in the streaming model.

\subsection{Our Results and Paper Outline}
Our main technical result in this paper are lower bounds of $\Omega(nk)$ randomized communication for the bitwise $k$-party AND, OR, and MAJ (majority) functions in the coordinator model. These sidestep clever upper bound techniques (e.g. Slepian-Wolf coding) and can be found in Section \ref{sec:bitwise}. In the same section we prove some lower bounds for AND and OR in the blackboard model as well. Back to the coordinator model, we show that the connectivity problem (given $k$ players with subgraphs on a common set of nodes, determine if it is connected) requires $\Omega(nk / \log^2 k)$ communication. This is in Section~\ref{sec:CONN}.

The coordinate-wise lower bounds imply lower bounds for the well-studied problems of distinct elements, $\eps$-approximate heavy-hitters, and $\eps$-kernels in the server-site model (or the other related models). We show any randomized algorithm requires at least $\Omega(nk)$, $\Omega(n/\eps)$ and $\Omega(k/\eps^{(d-1)/2})$ communication, respectively. The latter is shown to be tight. This is in Section~\ref{sec:application}.

We give some direct-sum-like results in Section \ref{sec:directsum}.

\subsection{Subsequent Work}
A series of work has been done after the conference version of this paper. Woodruff and Zhang~\cite{WZ12,WZ14} combined the symmetrization technique and a new technique called {\em composition} to show strong lower bounds for approximately computing a number of statistical problems, including distinct elements and frequency moments, in the coordinator model. The same authors also used symmetrization to prove tight lower bounds for the exact computation of a number of graph and statistical problems  \cite{WZ13}, and shaved a $\log k$ factor for the connectivity problem studied in this paper. The other $\log k$ factor in the connectivity lower bound can be further shaved using a relaxation of symmetrization proposed in \cite{WZ14}. 
As will be discussed  in Section~\ref{sec:limitations}, due to a limitation of the symmetrization technique, it cannot be used to prove a tight lower bound for the $k$-player disjointness problem. This was listed as an open problem in the conference version of this paper, and was later settled by Braverman et al.~\cite{BEOPV13}, using a different technique based on information complexity. Recently Huang et al.~\cite{HRVZ13} applied symmetrization together with information complexity to prove tight lower bounds for approximate maximum matchings; Li et al.~\cite{LSWW14} used the symmetrization technique to prove lower bounds for numerical linear algebra problems.
In another recent work, Chattopadhyay et al.~\cite{CRR14} further extended the symmetrization technique to general communication topology compared with the coordinator model which essentially has a star communication topology.

\section{Preliminaries}
\label{sec:pre}

In this section we review some basic concepts and definitions. We denote $[n]=\{1,\ldots,n\}$. All logarithms are base-2 unless noted otherwise.

\paragraph{Communication complexity.}
Consider two players Alice and Bob, given bit vectors $A$ and $B$,
respectively.  Communication complexity (see for example the book
\cite{KN97}) bounds the communication between Alice and Bob that is needed to compute
some function $f(A,B)$. The \emph{communication complexity} of a particular protocol is the maximal number of bits that it communicated, taken in the worst case over all pairs of inputs. The \emph{communication complexity} of the problem $f$ is the best communication complexity of $P$, taken over all protocols $P$ that correctly compute $f$.

Certain functions (such as $f = \EQ$ which determines if $A$ equals
$B$) can be computed with less communication if randomness is
permitted.  Let $R^\eps(f)$ denotes the communication complexity when the protocol is allowed to make a mistake with probability $\eps$. The error is taken over the randomness used by the protocol.

Sometimes we are interested in the case where the input of Alice and Bob is drawn from some distribution $\mu$ over pairs of inputs. We want to allow an error $\eps$, this time taken over the choice of the input. The worst-case communication complexity in this case is denoted by $D_\mu^\eps(f)$. Yao~\cite{Yao77} showed that $R^\eps(f) \ge \max_{\mu} D_\mu^\eps(f)$, thus in order to prove a lower bounds for randomized protocols it suffices to find a hard distribution and prove a distributional lower bound for it. This is called the Yao Minimax Principle.

In this paper we use an uncommon notion of \emph{expected distributional communication complexity}. In this case we consider the distributional setting as in the last paragraph, but this time consider the expected cost of the protocol, rather than the worst-case cost; again, the expectation is taken over the choice of input. We denote this $\ED_\mu^\eps(f)$.

\subsection{Two-Party Lower Bounds}
We state a couple of simple two-party lower bounds that will be useful in our reductions.

\paragraph{2-\BITS.}
\label{sec:2-BITS}
Let $\zeta_\rho$ be a distribution over bit-vectors of length $n$, where each bit is $1$ with probability $\rho$ and $0$ with probability $1-\rho$. In this problem Alice gets a vector drawn from $\zeta_\rho$, Bob gets a subset $S$ of $[n]$ of cardinality $\abs{S}=m$, and Bob wishes to learn the bits of Alice indexed by $S$.

The proof of the following lemma is in Appendix \ref{sec:2-BITS-proof}.

\begin{lemma}
\label{lem:2-BITS}
$\ED_{\zeta_\rho}^{1/3}(\textrm{2-\BITS}) = \Omega(n \rho \log(1/\rho))$.
\end{lemma}

\paragraph{2-\DISJ.}
\label{sec:2-DISJ}
In this problem Alice and Bob each have an $n$-bit vector. If we view vectors as sets, then each of them has a subset of $[n]$ corresponding to the $1$ bits. Let $x$ be the set of Alice and $y$ be the set of Bob. It is promised that $\abs{x \cap y} =1 \mbox{ or } 0$. The goal is to return $1$ if $x \cap y \neq \emptyset$, and $0$ otherwise.

\bigskip

We define the input distribution $\mu$ as follows. Let $l = (n+1)/4$. With probability $1/t$, $x$ and $y$ are random subsets of $[n]$ such that $\abs{x} = \abs{y} = l$ and $\abs{x \cap y} = 1$. And with probability $1 - 1/t$, $x$ and $y$ are random subsets of $[n]$ such that $\abs{x} = \abs{y} = l$ and $x \cap y = \emptyset$.  Razborov~\cite{Raz90} (see also
% There is an earlier paper by Kalyanasundaram and Schintger~
\cite{KS92}) proved that for $t=4$, $D^{1/100t}_{\mu}(\mbox{2-\DISJ}) = \Omega(n)$. In the following theorem we extend this result to general $t$ and also to the expected communication complexity. In Section \ref{sec:AND} we only need $t=4$, and in Section~\ref{sec:CONN} we will need general $t$.

The proof for the following lemma is in Appendix \ref{sec:2-DISJ-proof}.

\begin{lemma}
\label{lem:2-DISJ}
When $\mu$ has $|x \cap y| = 1$ with probability $1/t$ then
$\ED^{1/100t}_{\mu}(\mbox{2-\DISJ}) = \Omega(n)$.
\end{lemma}

\section{Bit-wise Problems}
\label{sec:bitwise}

\subsection{Multiparty AND/OR}
\label{sec:AND}

We now consider multiparty AND/OR (below we use $k$-\AND
and $k$-\OR for short). In the $k$-\AND problem, each player $i \ (1 \le i
\le k)$ has an $n$-bit vector $I_i$ and we want to establish the bitwise \AND of $I_i$, that is, $f_j(I_1, \ldots, I_k) = \bigwedge_i I_{i,j}$ for $j =
\{1,\ldots,n\}$. $k$-\OR is similarly defined with OR. Observe that the two
problems are isomorphic by $f_j(I_1, \ldots, I_k) =
\neg g_j(\bar{I_1}, \ldots, \bar{I_k})$ for $j = \{1, \ldots, n\}$, where
$\bar{I_i}$ is obtained by flipping all bits of $I_i$ and $g_j(\bar{I_1}, \ldots, \bar{I_k}) = \bigvee_i \bar{I}_{i,j}$. Therefore we only
need to consider one of them. Here we discuss $k$-\OR.

\subsubsection{Idea for the $k$-OR Lower Bound in the Coordinator Model}
\label{sec:AND:idea}

We now discuss the hard distribution and sketch how to apply the symmetrization technique for the $k$-OR problem in the coordinator model. The formal proof can be found in the next subsection.

We in fact start by describing two candidate hard distributions that \emph{do not} work. The reasons they do not work are interesting in themselves. Throughout this subsection, assume for simplicity that $k \ge 100 \log n$.

The most natural candidate for a hard distribution is to make each entry equal to $1$ with probability $1/k$. This has the effect of having each bit in the output vector be roughly balanced, which seems suitable for being a hard case. This is indeed the hard distribution for the blackboard model, but for the coordinator model (or the message-passing model) it is an easy distribution: Each player can send his entire input to the coordinator, and the coordinator can figure out the answer. The entropy of each player's input is only $\Theta((n \log k) / k)$, so the total communication would be $\Theta(n \log k)$ in expectation using e.g.\ Shannon's coding theorem;\footnote{To show the upper bounds in this subsection we use some notions from information theory without giving complete background for them. The reader can refer to e.g.~\cite{ThomasCover}, or alternatively can skip them entirely, as they are inconsequential for the remained of the paper and for understanding the symmetrization technique.} this is much smaller than the lower bound we wish to prove. Clearly, we must choose a distribution where each player's input has entropy $\Omega(n)$. This is the first indication that the $k$-player problem is significantly different than the $2$-player problem, where the above distribution is indeed the hard distribution.

The next candidate hard distribution is to randomly partition the $n$ coordinates into two equal-sized sets: The \emph{important set}, where each entry is equal to $1$ with probability $1/k$, and the \emph{balancing set}, where all entries are equal to $1$. Now the entropy of each player's input is $\Theta(n)$, and the distribution seems like a good candidate, but there is a surprising upper bound for this distribution: the coordinator asks $100\log n$ players to send him their entire input, and from this can easily figure out which coordinates are in the balancing set and which are in the important set. Henceforth, the coordinator knows this information, and only needs to learn the players' values in the important set, which again have low entropy. We would want the players to send these values, but the players themselves do not know which coordinates are in the important set, and the coordinator would need to send $nk$ bits to tell all of them this information. However, they do not need to know this in order to get all the information across: using a protocol known as \emph{Slepian-Wolf coding} (see e.g.\ \cite{ThomasCover}) the players can transmit to the coordinator all of their values in the important coordinates, with only $n \log k$ total communication (and a small probability of error). The idea is roughly as follows: each player $p_i$ chooses $100 n \log k / k$ sets $S_{i,j} \subseteq [n]$ independently and uniformly at random from public randomness. For each $j$, the player XORs the bits of his input in the coordinates of $S_{i,j}$, and sends the value of this XOR to the coordinator. The coordinator already knows the balancing set, so he only has $\Theta(n \log k / k)$ bits of uncertainty about player $i$'s input, meaning he can reconstruct the player's input with high probability (say by exhaustive search). The upper bound follows.

To get an actual hard distribution, we modify the hard distribution from the last paragraph. We randomly partition the $n$ coordinates into two equal-sized sets: The \emph{important set}, where each entry is equal to $1$ with probability $1/n$, and the \emph{noise set}, where each entry is equal to $1$ with probability $1/2$.\footnote{We could have chosen each entry in the important set to be equal to $1$ with probability $1/k$ as well, but choosing a value of $1/n$ makes the proofs easier. The important thing is to choose each of the noise bits to be equal to $1$ with probability $1/2$.} Clearly, each player's input has entropy $\Theta(n)$. Furthermore, the coordinator can again cheaply figure out which coordinates are in the important set and which are in the noise set, but the players do not know this information, and nothing like Slepian-Wolf coding exists to help them transmit the information to the coordinator. The distribution that we use in our formal proof is a little different than this, for technical reasons, but this distribution is hard as well.

We now sketch how to apply the symmetrization technique to prove that this distribution is indeed hard. To apply the symmetrization technique, we imagine giving Alice the input of one of the players, and to Bob the input of all of the others. Bob plays the coordinator, so Bob needs to compute the output; the goal is to prove a lower bound of $\Omega(n)$ on the communication complexity between Alice and Bob. What can Bob deduce about the answer? He can immediately take the OR of all the vectors that he receives, getting a vector where with good probability all of the noise bits are equal to $1$. (Recall that we assumed that the number of players is $k \ge 100 \log n$.) Among the important bits, roughly one of Alice's bits is equal to $1$, and the goal is to discover which bit it is. Alice cannot know which coordinates are important, and in essence we are trying to solve a problem very similar to Set Disjointness. A lower bound of $\Omega(n)$ is easy to get convinced of, since it is similar to the known lower bounds for set disjointness. Proving it is not entirely trivial, and requires making some small modifications to Razborov's classical lower bound on set disjointness~\cite{Raz90}.

We see that the lower bound for this distribution has to (implicitly) rule out a Slepian-Wolf type upper bound. This provides some evidence that any lower bound for the $k$-OR problem would have to be non-trivial.

\subsubsection{The Proof}

We prove the lower bound on $k$-\OR by performing a reduction from the promise version of the two-party {\em set disjointness} problem ($2$-\DISJ) which we lower-bounded in Lemma \ref{lem:2-DISJ}. Given an $(x,y)$ for $2$-\DISJ drawn from the hard distribution $\mu$, we construct an input for $k$-\OR.  Note that our mapping is not necessarily one-to-one, that is why we need a lower bound on the \emph{expected} distributional communication complexity of $2$-\DISJ.

\paragraph{Reduction.}
We start with Alice's input set $x$ and Bob's input set $y$ from the distribution $\mu$, with $t = 4$.  That is, both $x$ and $y$ have $l = n/4$ $1$ bits chosen at random under the condition that they intersect at one point with probability $1/4$, otherwise they intersect at no point.
For the $2$-\DISJ problem, we construct $k$ players' input sets $I_1, \ldots, I_k$ as follows. Let $z = [n] - y$. Let $S_2^l, \ldots, S_k^l$ be random subsets of size $l$ from $z$, and $S_2^{l-1}, \ldots, S_k^{l-1}$ be random subsets of size $l-1$ from $z$. Let $S_2^1, \ldots, S_k^1$ be random elements from $y$.
\begin{eqnarray*}
\left\{
\begin{array}{l}
I_1  =  x \\
I_{j} \ (j = 2, \ldots, k)  =  \left\{
  \begin{array}{rl}
    S_j^l & \textrm{w.p. } 1-1/4\\
    S_j^{l-1} \cup S_j^{1} & \textrm{w.p. } 1/4
  \end{array}
  \right.
\end{array}
\right.
\end{eqnarray*}
Let $\mu'$ be this input distribution for
$k$-\OR. If $I_j\ (2 \le j \le k)$ contains an element $S_j^1$, then we call this
element a special element.
\smallskip

This reduction can be interpreted as follows: Alice simulates a random
player $I_1$, and Bob, playing as the coordinator, simulates all the other $k-1$
players $I_2, \ldots, I_k$. Bob also keeps a set $V$ containing all the
special elements that ever appear in some $I_j\ (j = 2, \ldots, k)$. It is easy
to observe the following fact.
\begin{lemma}
\label{lem:AND-symmetric}
All $I_j\ (j = 1, \ldots, k)$ are chosen from the same distribution.
\end{lemma}
\begin{proof}
Since by definition $I_j\ (j = 2, \ldots, k)$ are chosen from the same distribution, we only need to show that $I_1 = x$ under $\mu$ is chosen from the same distribution as any $I_j$ is under $\mu'$.  Given $y$, note that $x$ is a random set of size $l$ in $z = [n]-y$ with probability $1/4$; and with the remaining probability $x$ is the union of a random set of size $l-1$ in $z$ along with a single random element from $y$.  This is precisely the distribution of each $I_j$ for $j\geq 2$.  
%According to the distribution $\mu'$, with probability $1/2$ the two random sets $x$ and $y$ are disjoint, and with probability $1/2$, they intersect at one element. This can be translated to the following: With probability $1/2$, $I_1$ is chosen as a random subset of size $l$ from $z = [n] - y$; and with the rest of the probability, $l-1$ elements in $I_1$ are chosen randomly from $z$ and the remaining $1$ element is chosen randomly from $y$.
\end{proof}

\noindent The following lemma shows the properties of our reduction.
\begin{lemma}
\label{lem:OR-DISJ}
If there exists a protocol $\cal P'$ for $k$-\OR on input distribution
$\mu'$ with communication complexity $C$ and error bound $\eps$, then
there exists a protocol $\cal P$ for the $2$-\DISJ on input distribution
$\mu$ with expected communication complexity $O(C/k)$ and error bound $\eps
+ 4k/n$.
\end{lemma}
\begin{proof}
Let us again view $I_j\ (j = 1, \ldots, k)$ as $n$-bit vectors. We show how
to construct a protocol $\cal P$ for $2$-\DISJ from a protocol $\cal P'$ for
$k$-\OR with the desired communication cost and error bound. $\cal P$ is
constructed as follows: Alice and Bob first run $\cal P'$ on $I_1, \ldots,
I_k$. Let $W \subseteq [n]$ be the set of indices where the results are
$1$. Bob checks whether there exists some $w \in W \cap y$ such that $w
\not\in V$. If yes, then $\cal P$ returns ``yes", otherwise $\cal P$
returns ``no".

We start by analyzing the communication cost of $\cal P$. Since player
$I_1$ is chosen randomly from the $k$ players, and
Lemma~\ref{lem:AND-symmetric} that all players' inputs are chosen from a
same distribution, the expected amount of communication between $I_1$
(simulated by Alice) and the other $k-1$ players (simulated by Bob) is at
most a $2/k$ fraction of the total communication cost of $\cal
P$. Therefore the expected communication cost of $\cal P$ is at most
$O(C/k)$.

For the error bound, we have the following claim: With probability at least
$(1 - 4k/n)$, there exists a $w \in W \cap y$ such that $w \not\in V$ if
and only if $x \cap y \not\in \emptyset$. First, if $x \cap y = \emptyset$,
then $I_1 = x$ cannot contain any element $w \in V \subseteq y$, thus the
resulting bits in $W \cap y$ cannot contain any special element that is not
in $V$. On the other hand, we have
$\Pr[((W \cap y) \subseteq V) \wedge (x \cap y \neq \emptyset)] \le
4k/n.$
This is because $((W \cap y) \subseteq V)$ and $(x \cap y \neq
\emptyset)$ hold simultaneously if and only if there exist some $S_j^1\ (1
\le j \le k)$ such that $S_j^1 \in x \cap y$. According to our random
choices of $S_j^1\ (j = 1, \ldots, k)$, this holds with probability at most
$k/l \le 4k/n$. Therefore if $\cal P'$ errors at most $\eps$, then $\cal P$
errors at most $\eps + 4k/n$.  
\end{proof}

\noindent Combining Lemma~\ref{lem:2-DISJ} and Lemma~\ref{lem:OR-DISJ},
we have the following theorem.
\begin{theorem}
\label{thm:k-OR}
$D_{\mu'}^{1/800}(\mbox{$k$-\OR}) = \Omega(nk)$, for $n \ge 3200k$ in the coordinator model.
\end{theorem}
\begin{proof}
If there exists a protocol $\cal P'$ that computes $k$-OR on input distribution $\mu'$ with communication complexity $o(nk)$ and error bound $1/800$, then by Lemma~\ref{lem:OR-DISJ} there exists a protocol $\cal P$ that computes $2$-\DISJ on input distribution $\mu$ with expected communication complexity $o(n)$ and error bound $1/800 + 4k/n \le 1/400$, contradicting Lemma~\ref{lem:2-DISJ} (when $t = 4$).  
\end{proof}

We discuss the applications of this lower bound to the distinct elements problem in Section \ref{sec:application}.

%%%%%%%%%%%%%%%%%%%%%%%%%%%%%%%%%%%%%%%%%%%%%%%%%

\subsection{Multiparty AND/OR with a Blackboard}
\label{sec:AND-blackboard}

Denote the $k$-OR problem in the blackboard model by $k$-\ORb. The general idea to prove a lower bound for $k$-\ORb is to perform
a reduction from a $2$-party bit-wise OR problem ($2$-\OR for short) with
public randomness. The $2$-\OR problem is the following.  Alice and
Bob each have an $n$-bit vector drawn from the following distribution:
Each bit is $1$ with probability $1/k$ and $0$ with
probability $1-1/k$. They also use public random bits to generate another
$(k-2)$ $n$-bit vectors such that each bit of these vectors is $1$ with
probability $1/k$ and $0$ with probability $1-1/k$. That is, the bits
are drawn from the same distribution as their private inputs. Let $\nu$ be
this input distribution. Alice and Bob want to compute bit-wise \OR of all
these $k$ $n$-bit vectors. For this problem we have the following theorem.
\begin{theorem}
\label{thm:2-OR}
$\ED^{1/3}_\nu(2\mbox{-\OR}) = \Omega(n/k \cdot \log k)$.
\end{theorem}

\begin{proof}
W.l.o.g., let us assume that Bob outputs the final result of $2$-\OR. It is
easy to see that if we take the bit-wise OR of the $(k-2)$ $n$-bit vectors
generated by public random bits and Bob's input vector, the resulting
$n$-bit vector $b$ will have at least a constant density of $0$ bits
with probability at least $1 - o(1)$, by a Chernoff bound. Since Bob can
see the $k-2$ public vectors, to compute the final result, all that Bob has
to know are bits of Alice's vector on those indices $i\ (1 \le i \le n)$
where $b[i] = 0$. In other words, with probability at least $1 - o(1)$,
Bob has to learn a specific, at least constant fraction of bits in Alice's vector up
to error $1/3$. Plugging in Lemma~\ref{lem:2-BITS} with $\rho = 1/k$, we know that the expected communication complexity is at least $(1 - o(1)) \cdot \Omega(n/k \cdot
\log k) = \Omega(n/k \cdot \log k)$.  
\end{proof}

We reduce this problem to $k$-\ORb as follows: Alice and Bob each
simulates a random player, and they use public random bits to simulate
the remaining $k-2$ players: The $(k-2)$ $n$-bit vectors generated by their
public random bits are used as inputs for the remaining $k-2$ random
players. Observe that the input of all the $k$ players are drawn from
the same distribution, thus the $k$ players are
symmetric. Consequently, the expected amount of communication between
the two random players simulated by Alice and Bob and other players
(including the communication between the two random players) is at
most $O(1/k)$ faction of the total communication cost of protocol for
$k$-\ORb.  We have the following theorem.

\begin{theorem}
\label{thm:k-OR-board}
$D^{1/3}_\nu(k\mbox{-\ORb}) = \Omega(n \cdot \log k)$.
\end{theorem}
\begin{proof}
It is easy to see that if we have a protocol for the $k$-\ORb on input
distribution $\nu$ with communication complexity $o(n \cdot \log k)$ and
error bound $1/3$, then we have a protocol for $2$-\OR on input distribution
$\nu$ with expected communication complexity $o(n/k \cdot \log k)$ and
error bound $1/3$, contradicting Theorem~\ref{thm:2-OR}.  
\end{proof}

It is easy to show a tight deterministic upper bound of $O(n \log k)$ for this problem in the blackboard model: each player speaks in turn and writes the coordinates where he has $1$ and no player that spoke before him had $1$.

\subsection{Majority}
\label{sec:MAJ}
In the $k$-\MAJ problem, we have $k$ players, each having a bit vector of
length $n$, and they want to compute bit-wise majority, i.e.\ to determine for each coordinate whether the majority of entries in this coordinate are $1$ or $0$. We prove a lower bound of $\Omega(nk)$ for this problem in the coordinator model by a reduction to 2-\BITS via symmetrization.

For the reduction, we consider $k = 2t+1$ players and describe the input
distribution $\tau$ as follows.  For each coordinate we
assign it either $t$ or $(t+1)$, each with probability $1/2$, independently over all coordinates. This indicates whether the $k$ players contain $t$ or $(t+1)$ $1$ bits among them in this coordinate, and hence whether that index has a majority of $0$ or $1$, respectively. Then we
place the either $t$ or $(t+1)$ $1$ bits randomly among the $k$
players inputs in this coordinate.  It follows that under $\tau$: (i) each player's bits are
drawn from the same distribution, (ii) each bit of each player is $1$
with probability $1/2$ and $0$ with probability $1/2$, and (iii) each
index has probability $1/2$ of having a majority of $1$s.

\begin{theorem}
\label{thm:k-MAJ}
$D^{1/6}_{\tau}(\mbox{k-\MAJ}) = \Omega(nk)$ in the coordinator model.
% for $n>13$.
\end{theorem}

\begin{proof}
Now we use symmetry to reduce to the two-player problem 2-\BITS.
Alice will simulate a random player under $\tau$ and Bob simulate the
other $k-1$ players.  Notice that, by (ii) any subset of $n'$ indices
of Alice's are from $\zeta_{1/2}$.  We will show that Alice and Bob
need to solve 2-\BITS on $\Omega(n)$ bits to solve $k$-\MAJ.  And, by
(i), since all players have the same distribution and Alice is a
random player, her expected cost in communicating to Bob is at most
$O(C/k)$ if $k$-\MAJ can be solved in $C$ communication.

Now consider the aggregate number of $1$ bits Bob has for each index;
he has $(t-1)$ $1$ bits with probability $1/4$, $t$ $1$ bits with
probability $1/2$, and $(t+1)$ $1$ bits with probability $1/4$.  Thus
for at most $(3/4)n$ indices (with probability at least $1-
\exp(-2(n/4)^2/n) = 1-\exp(-n/8) \geq 1-1/7$) that have either $(t-1)$
or $(t+1)$ $1$ bits Bob knows that these either will or will not have
a majority of $1$ bits, respectively.  But for the other at least $n/4
= \Omega(n)$ remaining indices for which Bob has exactly $t$ $1$ bits,
whether or not these indices have a majority of $1$ bits depends on
Alice's bit.  And conditioned on this situation, each of Alice's
relevant bits are $1$ with probability $1/2$ and $0$ with probability
$1/2$, hence distributed by $\zeta_{1/2}$.  Thus conditioned on at
least $n/4$ undecided indices, this is precisely the 2-\BITS problem
between Alice and Bob of size $n/4$.

Thus a protocol for $k$-\MAJ in $o(nk)$ communication and error bound
$1/6$ would yield a protocol for 2-\BITS in expected $o(n)$
communication and error bound $1/6 + 1/7 < 1/3$, by running the protocol simulated
between Alice and Bob.
% By repeating this simulated protocol some
% large constant number of times and returning the outcome that resolves
% a majority number of times, would result in a protocol for 2-\BITS
% with error bound $1/3$ and $o(n)$ communication.
This contradicts Lemma \ref{lem:2-BITS}, and proves that $k$-\MAJ requires $\Omega(kn)$ communication when allowing error on at most $1/6$ fraction of inputs.  
\end{proof}

\paragraph{Extensions.}
This lower bound can easily be extended beyond just majority (threshold $1/2$) to any constant threshold $\phi (0 < \phi < 1)$, by assigning to each coordinate either $\lfloor k\phi
\rfloor$ or $(\lfloor k\phi \rfloor+1)$ $1$ bits with probability
$1/2$ each. Let $\tau_{\phi}$ denote this distribution. Then the analysis just uses $\zeta_\phi$ in place of
$\zeta_{1/2}$, which also yields an $\Omega(n)$ lower bound for
2-\BITS. We call this extended $k$-\MAJ problem $(k,\phi)$-\MAJ.

\begin{corollary}
\label{cor:k-MAJ}
$D^{1/6}_{\tau_{\phi}}((k,\phi)\mbox{-\MAJ}) = \Omega(nk)$ for any constant $\phi (0 < \phi < 1)$ in the coordinator model.
% for $n>13$.
\end{corollary}

We discuss the applications of this lower bound to the heavy-hitter problem in Section \ref{sec:application}.

%\jeff{Another generalization is count-k-\MAJ where we want to know if
%  more than some number $\gamma < k$ bits have a majority.  I think
%  setting $\gamma = k/2$ in the above distribution would give us
%  something like a $\Omega(nk/\log k)$ lower bound.  Setting $\tau$
%  so exactly $k/2$ or $k/2+1$ indices have a majority, might allow us
%  to get an $\Omega(nk)$ lower bound.  But need to think through some
%  conditional probabilities in Fact 3.1.  }   

\section{Graph Connectivity}
\label{sec:CONN}

In the $k$-\CONN problem, we have $k$ players, each having a set of edges in an $n$-vertex graph. The goal is to decide whether the graph consisting of the union of all of these edges is connected. In this section we prove an $\Omega(nk/\log^2 k)$ lower bound for $k$-\CONN in the coordinator model, by performing a symmetry-based reduction from $2$-\DISJ.

\subsection{Proof Idea and the Hard Distribution}

Let us start by discussing the hard distribution. Assume for simplicity $k \ge 100 \log n$. We describe a hard distribution which is not quite the same as the one in the proof (due to technical reasons), but is conceptually clearer. In this hard distribution, we consider a graph $G$, which consists of two disjoint cliques of size $n/2$. We also consider one edge between these two cliques, called the \emph{connector}; the connector is not part of $G$. Each player gets as input $n/10$ edges randomly and uniformly chosen from the graph $G$; furthermore, with probability $1/2$ we choose exactly one random edge in one random player's input, and replace it by the connector. It is easy to see that if one of the players got the connector, then with high probability the resulting set of edges span a connected graph, and otherwise the graph is not connected.

To get convinced that the lower bound holds, notice that the coordinator can easily reconstruct the graph $G$. However, in order to find out if one of the players has received the connector, the coordinator needs to speak with each of the players to find this out. The situation is roughly analogous to the situation in the $k$-OR problem, since the players themselves did not get enough information to know $G$, and no Slepian-Wolf type protocol is possible since the edges received by each player are random-looking. The actual distribution that we use is somewhat more structured than this, in order to allow an easier reduction to the $2$-player disjointness problem.

\subsection{The Proof}
\label{sec:CONN-proof}
We first recall $2$-\DISJ. Similarly to before, in $2$-\DISJ, Alice and Bob have inputs $x\ (\abs{x} = \ell)$ and $y\ (\abs{y} = \ell)$ chosen uniformly at random from $[n]\ (n = 4\ell-1)$, with the promise that with probability $1/10k$, $\abs{x \cap y} = 1$ and with
probability $1 - 1/10k$, $\abs{x \cap y} = 0$. Let $\varphi$ be this
input distribution for $2$-\DISJ. Now given an input $(x,y)$ for
$2$-\DISJ, we construct an input for $k$-\CONN.

Let $K_{2n} = (V,E)$ be the complete graph with $2n$ vertices. Given Alice's input $x$ and Bob's input $y$, we construct $k$ players' input $I_1, \ldots, I_k$ such that $\abs{I_j} = \ell$ and $I_j \subseteq E$ for all $1 \le j \le k$.
We first pick a random permutation $\sigma$ of $[2n]$.
Alice constructs $I_1 = \{(\sigma(2i-1), \sigma(2i))\ |\ i \in x\}.$

Bob constructs $I_2, \ldots, I_k$.
It is convenient to use $\sigma$ and $y$ to divide $V$ into two subsets $L$ and $R$.  For each $i\ (1 \le i \le n)$, if $i \in y$, then with probability $1/2$, we add $\sigma(2i-1)$ to $L$ and $\sigma(2i)$ to $R$; and with the rest of the probability, we add
$\sigma(2i-1)$ to $R$ and $\sigma(2i)$ to $L$. Otherwise if $i \not\in y$, then with probability $1/2$, we add both $\sigma(2i-1)$ and $\sigma(2i)$ to $L$; and with the rest of the probability, we add both $\sigma(2i-1)$ and $\sigma(2i)$ to $R$. Let $K_L = (L, E_L)$ and $K_R = (R, E_R)$ be the two complete graphs on sets of vertices $L$ and
$R$, respectively.
Now using $E_R$ and $E_L$, Bob can construct each $I_j$.
With probability $1-1/10k$, $I_j$ is a random subset of disjoint edges (i.e. a matching) from $E_L \cup E_R$ of size $\ell$;  and with probability $1/10k$, $I_j$ is a random subset of disjoint edges from $E_L \cup E_R$ of size $\ell -1$ and one random edge from $E \setminus (E_L \cup E_R)$.

%Let $M_2^{\ell}, \ldots, M_k^{\ell}$ be subsets of
%{\em disjoint edges} (i.e., matchings) of size $\ell$ picked uniformly
%at random from $E_L \cup E_R$, and $M_2^{\ell-1}, \ldots,
%M_k^{\ell-1}$ be subsets of {\em disjoint edges} of size $\ell-1$
%picked uniformly at random from $E_L \cup E_R$. Let $M_2^1, \ldots,
%M_k^1$ be edges picked uniformly at random from $E \backslash (E_L
%\cup E_R)$, with the constraints that $M_j^1$ is disjoint from all
%edges in $M_j^{\ell-1}$ for each $j\ (2 \le j \le k)$.
%\[
%I_{j} \ (j = 2, \ldots, k)  =  \left\{
%  \begin{array}{rl}
%    M_j^{\ell} & \text{w.p. } 1-1/\alpha k \\
%    M_j^{\ell-1} \cup M_j^{1} & \text{w.p. } 1/\alpha k.
%  \end{array}
%  \right.
%\]
Let $\varphi'$ be the input distribution for $k$-\CONN defined as above for each $I_j\ (1 \le j \le k)$.
We define the following two events.
\begin{itemize}
\item[$\xi_1$:] Both edge-induced subgraphs $\cup_{j=2}^k I_j \bigcap E_L$ and $\cup_{j=2}^k I_j \bigcap E_R$ are connected, and span $L$ and $R$ respectively.
\item[$\xi_2$:] $\cup_{j=2}^k I_j \bigcap E$ is {\em not} connected.
\end{itemize}
It is easy to observe the following two facts by our construction.
\begin{lemma}
\label{lem:symmetric-2}
All $I_j\ (j = 1, \ldots, k)$ are chosen from the same distribution.
\end{lemma}
\begin{proof}
Since $\sigma$ is a random permutation of $[2n]$, according to the
distribution $\varphi'$, Alice's input can also seen as follows: With
probability $1 - 1/10k$, it is a random matching of size $\ell$ from $E_L
\cup E_R$; and with probability $1/10k$, it is a matching consists of
$\ell-1$ random edges from $E_L \cup E_R$ and one random edge from $E
\backslash (E_L \cup E_R)$, which is $(\sigma(2z-1), \sigma(2z))$ where $z
= x \cap y$.  
\end{proof}

\begin{lemma}
\label{lem:connected}
$\xi_1$ happens with probability at least $1-1/2n$ when $k \geq 68 \ln n + 1$.
\end{lemma}

%\qin{Sketched proof will be used in the 10-page abstract}
%\begin{proof} (This is a proof sketch; full proof in Appendix \ref{sec:CONN-proof}.)
%First note that by our construction, both $\abs{L}$ and $\abs{R}$ are $\Omega(n)$ with high probability. To locally simplify notation, we consider a graph $(V,E)$ of $n$ nodes
%where edges are drawn in $(k - 1) \ge 68\ln n$ rounds, and each round $n/4$ disjoint
%edges are added to the graph.  If $(V,E)$ is connected with
%probability $(1-1/4n)$, then by union bound over $\cup_{j=2}^k I_j
%\bigcap E_L$ and $\cup_{j=2}^k I_j \bigcap E_R$, $\xi_1$ is true with
%probability $(1-1/2n)$. The proof follows four steps.
%\begin{itemize}
%\item[(S1)] We show that all points have degree at least $8\ln n$ with probability at least $1 - 1/12n$; this uses the first $28 \ln n$ rounds.
%\item[(S2)] We show (conditioned on (S1)) that any subset $S \subset V$ of $h < n/10$ points is connected to at least $\min\{h \ln n, n/10\}$ distinct points in $V \setminus S$, with probability at least $1-1/n^2$.
%\item[(S3)] We can iterate (S2) $\ln n$ times to show that there must be a single connected component $S_G$ of size at least $n/10$, with probability at least $1-1/12n$.
%\item[(S4)] We can show (conditioned on (S3) and using the last $40 \ln n$ rounds) that all points are connected to $S_G$ with probability at least $1-1/12n$.  \qedhere
%\end{itemize}
%\end{proof}

\begin{proof}
 (This is a proof sketch; full proof in Appendix \ref{sec:CONN-proof-appendix}.)
First note that by our construction, both $\abs{L}$ and $\abs{R}$ are $\Omega(n)$ with high probability. To locally simplify notation, we consider a graph $(V,E)$ of $n$ nodes
where edges are drawn in $(k - 1) \ge 68\ln n$ rounds, and each round $n/4$ disjoint
edges are added to the graph.  If $(V,E)$ is connected with
probability at least $(1-1/4n)$, then by union bound over $\cup_{j=2}^k I_j
\bigcap E_L$ and $\cup_{j=2}^k I_j \bigcap E_R$, $\xi_1$ is true with
probability at least $(1-1/2n)$. The proof follows four steps.
\begin{itemize}
\item[(S1)] Using the first $28 \ln n$ rounds, we can show that all vertices have degree at least $8\ln n$ with probability at least $1 - 1/12n$.
\item[(S2)] Conditioned on (S1), any subset $S \subset V$ of $h < n/10$ vertices is connected to at least $\min\{h \ln n, n/10\}$ distinct vertices in $V \setminus S$, with probability at least $1-1/n^2$.
\item[(S3)] Iterate (S2) $\ln n$ times to show that there must be a single connected component $S_G$ of size at least $n/10$, with probability at least $1-1/12n$.
\item[(S4)] Conditioned on (S3), using the last $40 \ln n$ rounds we can show that all vertices are connected to $S_G$ with probability at least $1-1/12n$.   \qedhere
\end{itemize}
\end{proof}

\noindent The following lemma shows the properties of our reduction.
\begin{lemma}
\label{lem:conn-reduction}
Assume $k \ge 100 \log n$. If there exists a protocol $\cal P'$ for $k$-\CONN on input distribution
$\varphi'$ with communication complexity $C$ and error bound $\eps$, then there exists a protocol $\cal P$ for the
$2$-\DISJ on input distribution $\varphi$ with expected communication
complexity $O(C/k \cdot \log k)$ and error bound $(29 \ln k \cdot \eps + 1/2000k)$.
\end{lemma}
\begin{proof}

In $\cal P$, Alice and Bob first construct $\{I_1, \ldots, I_k\}$ according to our reduction, and then run the protocol $\cal P'$ on it. By Lemma~\ref{lem:connected} we have that $\xi_1$
holds with probability at least $1 - 1/2n$. And by our
construction, conditioned on that $\xi_1$ holds, $\xi_2$
holds with probability at least $(1-1/10k)^{k-1} \ge 1 - 1/10$. Thus the input generated by a random reduction encodes the $2$-\DISJ problem
with probability at least $(1 - 1/2n - 1/10) > 8/9$. We call such an input a {\em good} input. We repeat the random reduction $c \ln k$ times (for some large enough constant $c$; e.g., $c\geq 29$) and run $\cal P'$ on each of the resulting inputs for $k$-\CONN. The probability that at least $2/3$-fraction of inputs are good is at least $1 - 1/2000k$ (by picking $c$ large enough; a Chernoff bound shows this probability is at most $2\exp(-2 (2/9)^2 (c \ln k))$ which is less than $1/(2000 k)$ with $c\geq 29$ and $k \geq 100$). Thus we can output the majority of the $c\ln k$ runs of $\cal P'$, and obtain a protocol $\cal P$ for $2$-\DISJ with expected communication complexity $O(C/k \cdot \log k)$ and error bound $\eps \cdot 29 \ln k + 1/2000k$; the $\eps \cdot 29 \ln k$ term comes from a union bound over the $\eps$ error on each of $29 \ln k$ runs of $k$-\CONN.
\end{proof}

\noindent Combining Lemma~\ref{lem:2-DISJ} and
Lemma~\ref{lem:conn-reduction}, we have the following theorem.
\begin{theorem}
$D_{\varphi'}^{1/(2000\cdot (29 \ln k) \cdot k)}(\mbox{$k$-\CONN}) = \Omega(nk/\log k)$, for  $k \ge 100 \log n$ in the coordinator model.
%$D_{\varphi'}^{1/3}(\mbox{$k$-\CONN}) = \Omega(nk/\log k)$
\end{theorem}
\begin{proof}
%If there exists a protocol $\cal P'$ that computes $k$-\CONN with communication complexity $o(nk/\log k)$ and error $1/3$, then by Lemma~\ref{lem:OR-DISJ} there exists a protocol $\cal P$ that computes $2$-\DISJ with expected communication complexity $o(n)$ and error $1/2000k (= 1/100\alpha k)$, contradicting Theorem~\ref{thm:2-DISJ} (when $t = \alpha k$).
If there exists a protocol $\cal P'$ that computes $k$-\CONN with
communication complexity $o(nk/\log k)$ and error $1/(2000\cdot (29 \ln k) \cdot k)$, then by
Lemma~\ref{lem:OR-DISJ} there exists a protocol $\cal P$ that computes
$2$-\DISJ with expected communication complexity $o(n)$ and error at most
 $(1/(2000\cdot (29 \ln k) \cdot k) \cdot (29 \ln k)) + 1/2000k = 1/1000k$, contradicting Lemma~\ref{lem:2-DISJ} (when $t = 10k$).  
\end{proof}

% Finally, we have $R^{1/3}(\mbox{$k$-\CONN}) \ge \Omega(D_{\varphi'}^{1/3}
% (\mbox{$k$-\CONN}))  \ge \Omega(nk/\log k)$ as an immediate consequence.

Finally, we have the following immediate consequence.
\begin{eqnarray*}
R^{1/3}(\mbox{$k$-\CONN})  
&\ge&  \Omega(R^{1/(2000\cdot (29 \ln k) \cdot k)}(\mbox{$k$-\CONN}) / \log k) \\ 
&\ge& \Omega(D_{\varphi'}^{1/(2000\cdot (29 \ln k) \cdot k)}(\mbox{$k$-\CONN})/\log k)  \\
&\ge& \Omega(nk/\log^2 k)
\end{eqnarray*}

%\smallskip
%The hardness of the graph connectivity problem may indicate that most graph problems in the coordinator model or message-passing model do not have efficient solutions other than trivial ones.

\section{Some Direct-Sum-Like Results}
\label{sec:directsum}
Let $f : {\cal X} \times {\cal Y} \to {\cal Z}$ be an arbitrary function. Let
$\mu$ be a probability distribution over ${\cal X} \times {\cal Y}$. Consider a setting where we have
$k+1$ players: Carol, and $P_1, P_2, \ldots, P_k$. Carol receives an input from $x \in {\cal X}$ and each $P_i$ receives an input $y_i \in {\cal Y}$. Let $R^{\eps}(f^k)$ denote the randomized communication complexity of computing
$f$ on Carol's input and each of the $k$ other players' inputs
respectively; i.e., computing $f(x, y_1), f(x, y_2), \ldots, f(x,
y_k)$. Our direct-sum-like theorem in the message-passing model states
the following.

\begin{theorem}
\label{thm:directsum}
In the message-passing model, for any function $f : {\cal X}
\times {\cal Y} \to {\cal Z}$ and any distribution $\mu$ on ${\cal
  X} \times {\cal Y}$, we have $R^{\eps}(f^k) \ge \Omega(k \cdot
\ED^{\eps}_{\mu}(f))$.
\end{theorem}

Note that this is not a direct-sum theorem in the strictest sense, since it relates randomized communication complexity to expected distributional complexity. However, it should probably be good enough for most applications.
The proof of this theorem is dramatically simpler than most direct-sum proofs known in the literature (e.g.~\cite{direct_sum}). This is not entirely surprising, as it is weaker than those theorems: it deals with the case where the inputs are spread out over many players, while in the classical direct-sum setting, the inputs are only spread out over two players (this would be analogous to allowing the players $P_i$ to communicate with each other for free, and only charging them for speaking to Carol). However, perhaps more surprisingly, \emph{optimal} direct-sum results are \emph{not known} for most models, and are considered to be central open questions, while the result above is essentially optimal. Optimal direct-sum results in $2$-player models would have dramatic consequences in complexity theory (see e.g.\
\cite{direct_sum_circuit} as well as \cite{patrascu_conj}), so it seems interesting to check whether direct-sum results in multiparty communication could suffice for achieving those complexity-theoretic implications.

\begin{proof}
Given the distribution $\mu$ on ${\cal X} \times {\cal Y}$, we
construct a distribution $\nu$ on ${\cal X} \times {\cal Y}^k$. Let
$\rho_x$ be the marginal distribution on ${\cal Y}$ induced by $\mu$
condition on $X = x$.  We first pick $(x, y_1) \in {\cal X} \times
{\cal Y}$ according to $\mu$, and then pick $y_2, \ldots, y_k$
independently from ${\cal Y}$ according to $\rho_x$ . We show that
$D_{\nu}^{\eps}(f^k) \ge \Omega(k \cdot \ED_{\mu}^{\eps}(f))$. The
theorem follows by Yao's min-max principle.

Suppose that Alice and Bob get inputs $(u, w)$ from ${\cal X} \times
{\cal Y}$ according to $\mu$. We can use a protocol for $f^k$ to
compute $f(u,w)$ as follows: Bob simulates a random player in
$\{P_1, \ldots, P_k\}$. W.l.o.g, say it is $P_1$. Alice simulates
Carol and the remaining $k-1$ players. The inputs for Carol and $P_1,
\ldots, P_k$ are constructed as follows: $x = u$, $y_1 = w$ and $y_2,
\ldots, y_k$ are picked from $\cal Y$ according to $\rho_u$ (Alice
knows $u$ and $\mu$ so she can compute $\rho_u$). Let $\nu$ be the
distribution of $(x, y_1, \ldots, y_k)$ in this construction. We now
run the protocol for $f^k$ on $x, y_1, \ldots, y_k$. The result also
gives $f(u, w)$.

Since $y_1, \ldots, y_k$ are chosen from the same distribution and
$P_1$ is picked uniformly at random from the $k$ players other than
Carol, we have that in expectation, the expected amount of
communication between $P_1$ and $\{$Carol, $P_2, \ldots, P_k \}$, or
equivalently, the communication between Alice and Bob according to our
construction, is at most a $2/k$ fraction of the total communication
of the $(k+1)$-player game. Thus $D_{\nu}^{\eps}(f^k) \ge \Omega(k \cdot
\ED_{\mu}^{\eps}(f))$.  
\end{proof}

\subsection{With Combining Functions}

In this section we extend Theorem~\ref{thm:directsum} to the
complexity of the AND/OR of $k$ copies of $0/1$ function
$f$. Similar to before, AND and OR are essentially the same so we only
talk about OR here. Let $R^{\eps}(f^k_{\mathrm{OR}})$ denote the
randomized communication complexity of computing $f(x, y_1) \vee
f(x, y_2) \vee \ldots \vee f(x, y_k)$.

\begin{theorem}
\label{thm:directsum-OR}
In the message-passing model, for any function $f : {\cal X}
\times {\cal Y} \to \{0,1\}$ and every distribution $\mu$ on ${\cal X}
\times {\cal Y}$ such that $\mu(f^{-1}(1)) \le 1/10k$, we have
$R^{1/3}(f^k_{\mathrm{OR}}) \ge \Omega(k/\log^2(1/\eps) \cdot
\ED^{\eps}_{\mu}(f))$.
\end{theorem}
\begin{proof}
 The
reduction is the same as that in the proof of
Theorem~\ref{thm:directsum}. Note that if $\mu(f^{-1}(1)) \le 1/10k$, then with probability $(1
- 1/10k)^{k-1} \ge 0.9$, we have $f^k_{\mathrm{OR}}(x, y_1, \ldots,
y_k) = f(u,w)$. Similar to the proof of Lemma~\ref{lem:conn-reduction}, we can repeat the reduction for $c \log (1/\eps)$ times for some large enough constant $c$, and then the majority of $f^k_{\mathrm{OR}}(x, y_1, \ldots,
y_k)$'s will be equal to $f(u,w)$ with probability $1 - \eps/2$.
Therefore if we have a protocol for $f^k_{\mathrm{OR}}$ with error probability $\eps/(2c \cdot \log 1/\eps)$ under $\nu$ and communication complexity $C$, then we have a
protocol for $f$ with error probability $(c \log 1/\eps) \cdot (\eps/(2c \log 1/\eps)) + \eps/2) = \eps$ under $\mu$ and  expected communication complexity
$O(\log(1/\eps) \cdot C/k)$.  Consequently, $R^{1/3}(f^k_{\mathrm{OR}}) \ge \Omega\left(R^{\eps/(2c \log 1/\eps)}(f^k_{\mathrm{OR}}) / \log(1/\eps)\right) \ge  \Omega\left(D^{\eps/(2c \log 1/\eps)}_\nu(f^k_{\mathrm{OR}}) / \log(1/\eps)\right) \ge \Omega\left(k/\log^2(1/\eps) \cdot
\ED^{\eps}_{\mu}(f)\right)$.
\end{proof}

%It would be interesting to generalize this to other combining functions such as majority.

%One can also use symmetrization to prove similar direct-sum results for other settings. One such setting is when there are $2k$ players: $k$ players that get $x_1,\ldots,x_k$ and $k$ players that get $y_1,\ldots,y_k$, and the goal is to compute $f(x_1, y_1),$ $f(x_2, y_2), \ldots, f(x_k, y_k)$. Another setting is the same except that there are only $k+1$ players, and Carol receives all of the inputs $x_1,\ldots,x_k$. We omit the details here. Versions of direct-sum in the blackboard model are also possible for some of these settings.

%\todo{A direct-sum theorem for MAJ?}

%\qin{It seems that if $f^{-1}(1) = \tau$, what we really can distinguish is $k\tau$ VS $k\tau +1$, instead of MAJ. But I don't think this is very interesting. What do you think?}

%\qin{I think we should put a remark here, pointing out the main difference between the direct-sum and the $k$-player communication complexity. Roughly speaking, in the the $k$-player communication  complexity, after choosing the $2$-player game to reduce from, we  have to make sure that the random $k-1$ players created by Alice can ``reconstruct" Alice's original input. More preciously, the union of  the random $k-1$ players' inputs together with the result of the  $k$-player game can tell a lot about Bob's input.}

%\section{Distributed Sensing Lowerbounds (of $\eps$-Kernels)}
\section{Applications}
\label{sec:application}

We now consider applications where multiparty communication complexity lower bounds such as ours are needed.  As mentioned in the introduction, our multiparty communication problems are strongly motivated by research on the server-site model and the more general distributed streaming model.
We discuss three problems here:
the heavy hitters problem, which asks to find the approximately most frequently occurring elements in a set which is distributed among many clients;
the distinct elements problems, which lists the distinct elements from a fixed domain where the elements are scattered across distributed databases possibly with multiplicity;
and the $\eps$-kernel problem, which asks to approximate the convex hull of a set which is distributed over many clients.

%%%%%%%%%%%%%%%%%%%%%%%%%%%%%
\paragraph{Distinct elements.}
Consider a domain of $n$ possible elements and $k$ distributed databases each which contains a subset of these elements.  The exact distinct elements problem is to list the set of all distinct elements from the union of all elements across all distributed databases.  This is precisely the $k$-OR problem and follows from Theorem \ref{thm:k-OR}, since the existence of each element in each distributed data point can be signified by a bit, and the bit-wise OR represents the set of distinct elements.

\begin{theorem}\label{thm:distinct-element}
For a set of $k$ distributed databases, each containing a subset of $n$ elements, it requires $\Omega(nk)$ communication total between the databases to list the set of distinct elements with probability at least $2/3$.
\end{theorem}

%%%%%%%%%%%%%%%%%%%%%%%%%%%%%
\paragraph{$\eps$-Kernels.}
 Given a set of $n$ points $P \subset \b{R}^d$, the width in direction
 $u$ is denoted by
\[
\wid(P,u) = \left(\max_{p \in P} \langle p, u \rangle\right) - \left(\min_{p \in P} \langle p, u \rangle \right),
\]
where $\langle \cdot, \cdot \rangle$ is the standard inner product operation.
Then an $\eps$-kernel~\cite{AHV04,AHV07}  $K$ is a subset of $P$ so that for any direction $u$ we have
\[
\wid(P,u) - \wid(K,u)  \leq \eps \cdot \wid(P,u).
\]
An $\eps$-kernel $K$ approximates the convex hull of a point set $P$, such that if the convex hull of $K$ is expanded in any direction by an $\eps$-factor it contains $P$.  As such, this coreset has proven useful in many applications in computational geometry such as approximating the diameter and smallest enclosing annulus of point sets~\cite{AHV04,AHV07}.  It has been shown that $\eps$-kernels may require $\Omega(1/\eps^{(d-1)/2})$ points (on a $(d-1)$-sphere in $\b{R}^d$) and can always be constructed of size $O(1/\eps^{(d-1)/2})$ in time $O(n + 1/\eps^{d-3/2})$~\cite{YAPV04,Cha06}.

We note a couple of other properties about $\eps$-kernels.
Composibility: If $K_1, \ldots, K_k$ are $\eps$-kernels of $P_1, \ldots, P_k$, respectively, then $K = \bigcup_{i=1}^k K_i$ is an $\eps$-kernel of $P = \bigcup_{i=1}^k P_i$.
Transitivity: If $K_1$ is an $\eps_1$-kernel of $P$ and $K_2$ is an $\eps_2$-kernel of $K_1$, then $K_2$ is an $(\eps_1 + \eps_2)$-kernel of $P$.
Thus it is easy to see that each site $i$ can simply send an $(\eps/2)$-kernel $K_i$ of its data of size $n_\eps = O(1/\eps^{(d-1)/2})$ to the server, and the server can then create and $(\eps/2)$-kernel of $\bigcup_{i=1}^k K_i$ of size $O(1/\eps^{(d-1)/2})$.  This is asymptotically the optimal size for and $\eps$-kernel of the full distributed data set.  We next show that this procedure is also asymptotically optimal in regards to communication.

\begin{theorem}\label{thm:eKerLB}
For a distributed set of $k$ sites, it requires $\Omega(k/\eps^{(d-1)/2})$ communication total between the sites and the server for the server to create an $\eps$-kernel of the distributed data with probability at least $2/3$.
\end{theorem}
\begin{proof}
We describe a construction which reduces $k$-OR to this problem, where each of $k$ players has $n_\eps = \Theta(1/\eps^{(d-1)/2})$ bits of information.  Theorem \ref{thm:k-OR} shows that this requires $\Omega(n_\eps k)$ communication.

We let each player have very similar data $P_i = \{p_{i,1}, \ldots, p_{i,n_\eps}\}$, each player's data points lie in a unit ball $B = \{q \in \b{R}^d \mid \|q\| \leq 1\}$.  For each player, their $n_\eps$ points are in similar position.  Each player's $j$th point $p_{i,j}$ is along the same direction $u_j$, and its magnitude is either $\|p_{i,j}\| = 1$ or $\|p_{i,j}\| = 1-2\eps$.  Furthermore, the set of directions $U = \{u_i\}$ are well-distributed such that for any player $i$, and any point $p_{i,j}$ that $P_i \setminus p_{i,j}$ is not an $\eps$-kernel of $P_i$; that is, the only $\eps$-kernel is the full set.  The existences of such a set follows from the known lower bound construction for size of an $\eps$-kernel.

We now claim that the $k$-OR problem where each player has $n_\eps$ bits can be solved by solving the distributed $\eps$-kernel problem under this construction.  Consider any instance of $k$-OR, and translate to the $\eps$-kernel problem as follows.  Let the $j$th point $p_{i,j}$ of the $i$th player have norm $\|p_{i,j}\| = 1$ when $j$th bit of the player is $1$, and have norm $\|p_{i,j}\| = 1-2\eps$ if the $j$th bit is $0$.
By construction, an $\eps$-kernel of the full set must acknowledge (and contain) the $j$th point from some player that has such a point with norm $1$, if one exists.  Thus the full $\eps$-kernel encodes the solution to the $k$-OR problem: it must have $n_\eps$ points and, independently, the $j$th point has norm $1$ if the $j$th OR bit is $1$, and has norm $1-2\eps$ if the $j$th OR bit is $0$.  
\end{proof}

%%%%%%%%%%%%%%%%%%%%%%%%%%%%%

%In Section~\ref{sec:MAJ}, we have shown that our lower bound problem almost matches the current best upper bound~\cite{HYLC11} for the heavy hitter problem. \qin{Shall we say this here or in Section~\ref{sec:MAJ}?} In this section, we will mainly focus on the problem of computing $\eps$-kernels in the server-site model.

\paragraph{Heavy Hitters.}
Given a multi-set $S$ that consists of $n$ elements, a threshold parameter $\phi$, and an error parameter $\eps$, the \emph{approximate heavy hitters} problem asks for a set of elements which contains all elements that occur at least $\phi k$ times in $S$ and contains no elements that occur fewer than $\phi k (1-\eps)$ times in $S$.  On a static non-distributed data set this can easily be done with sorting.  This problem has been famously studied in the streaming literature where the Misra-Gries~\cite{MG82} and SpaceSaving~\cite{MAA06} summaries can solve the problem in optimal space.  In the distributed setting the best known algorithms use random sampling of the indices and require either $O((1/\eps^2)n \log n)$ or $O(k + \sqrt{k} n/\eps \cdot \log n)$ communication to guarantee a correct set with constant probability~\cite{HYLC11}.  
We will prove a lower bound of $\Omega(n/\eps)$ here. After our work Woodruff and Zhang~\cite{WZ12} showed a lower bound of $\Omega(\min\{n/\eps^2, \sqrt{k}n/\eps\})$ (translating to our setting; their setting is a bit different), which is tight up to a $\log$ factor. 
% and leave tightening this bound as an open problem.
%\jeff{is this now closed?}

We now present a specific formulation of the approximate heavy hitters problem as $(k,\phi,\eps)$-\HH as follows.  Consider $k$ players, each with a bit sequence (either 0 or 1) of length $n$ where each coordinate represents an element.  The goal is to answer YES for each index with at least $\phi k$ elements, NO for each index with no more than $\phi k(1-\eps)$ elements, and either YES or NO for any count in between.
%
%We can also show a lower bound for a multi-party heavy-hitters problem.  In this setting we want to distinguish the number of bits at each index from at least $k \phi$ (YES) with at most $k \phi (1-\eps)$ (NO).  Any number in between we are free to answer YES or NO.  On a multiparty setting where $k$ players each have either 0 or 1 bit at each index, we label this problem $(k,\eps)$-\HH.
%
The reduction is based on a distribution $\tau_{\phi,\eps}$ where independently each index has either $\phi k$ or $\phi k(1-\eps)$ 1 bits, each with probability $1/2$.  In the reduction the players are grouped into sets of $k \eps$ players each, and all grouped players for each index are either given a 1 bit or all players are given a 0 bit.  These 1 bits are distributed randomly among the $1/\eps$ groups.  The proof then uses Corollary \ref{cor:k-MAJ}.

\begin{theorem}
$D_{\tau_{\phi,\eps}}^{1/6}((k,\phi,\eps)$-$\HH) = \Omega(n/\eps)$.
\end{theorem}
\begin{proof}
To lowerbound the communication for $(k,\phi,\eps)$-$\HH$ we first show another problem is hard: $(1/\eps,\phi)$-$\HH$ (assume $1/\eps$ is an integer).  Here there are only $1/\eps$ players, and each player at each index has a count of $0$ or $k\eps$, and we again want to distinguish between total counts of at least $k \phi$ (YES) and at most $k \phi(1-\eps)$ (NO).  By distribution $\tau_{\phi,\eps}$ each index has a total of either $k \phi$ or $k \phi(1-\eps)$ exactly.  And then we distribute these bits to players so each player has precisely either $0$ or $k\eps$ at each index.  When $k$ is odd, this is
precisely the $(1/\eps,\phi)$-MAJ problem, which by Corollary \ref{cor:k-MAJ} takes $\Omega(n/\eps)$ communication.

Now it is easy to see that $D_{\tau_{\phi,\eps}}^{1/6}((1/\eps,\phi)$-$\HH) \leq D_{\tau_{\phi,\eps}}^{1/6}((k,\phi,\eps)$-$\HH)$, since the former on the same input allows $(1/\eps)$ sets of $k\eps$ players to talk to each other at no cost.  
\end{proof}

\section{Concluding Remarks}
\label{sec:conclusion}

In this paper we have introduced the symmetrization technique, and have shown how to use it to prove lower bounds for $k$-player communication games. This technique seems widely applicable, and we expect future work to find further uses.

\subsection{A Brief Comparison to the icost Method}
In this section we make a brief comparison between our symmetrization method and the celebrated {\em icost} method by~\cite{BYJKS02}. Readers who are familiar with the icost method may notice that the $k$-\XOR, $k$-\MAJ and blackboard $k$-\OR/\AND problems discussed in this paper can also be handled by the icost method. However, for problems whose complexities are different in the blackboard model and the message-passing model, e.g., $k$-\OR and $k$-\CONN, the icost method cannot be used to obtain tight lower bounds in the message-passing/coordinator model, while the symmetrization method still applies.

If we view the input to $k$ players as a matrix with players as rows each having an $n$-bit input, the icost method first ``divides" the whole problem to $n$ copies of primitive problems column-wise, and then analyzes a single primitive problem. While the symmetrization method first reduces the size of a problem in the row space, that is, it first reduce a $k$-player problem to a $2$-player problem, and then analyze the $2$-player problem. We can certainly use the icost method again when analyzing the resulting $2$-player problem, which gives us an elegant way to combine these two techniques.

We notice that after the conference version of this paper, Braverman et al.~\cite{BEOPV13} and Huang et al.~\cite{HRVZ13} independently developed two new (and different) definitions for icost in the coordinator model, and used them to prove some tight lower bounds in the coordinator model.

\subsection{Limitations and Future Directions}
\label{sec:limitations}
The symmetrization technique also has several limitations, which we wish to discuss here.

Firstly, there are problems that might be impossible to lower bound using symmetrization. Consider for example the $k$-player disjointness problem, where each player gets a subset of $\{1,\ldots,n\}$, and the goal is to decide whether there is an element that appears in all of the sets. This problem looks to be easier than the coordinate-wise AND problem. But in the conference version of this paper we conjectured that this problem has a communication lower bound of $\Omega(nk)$ in the coordinator model as well. However, it seems impossible to prove this lower bound using symmetrization, for the following reason. Suppose we give Alice the input of a randomly-chosen player $p_i$, and give Bob the inputs of all the other players. It seems that for any symmetric input distribution, this problem can be solved using $O((n \log k) / k)$ bits in expectation, which is much lower than the $\Omega(nk)$ lower bound we are aiming for. 
%It is not out of the question that some variant of the symmetrization technique can be used to get the $\Omega(nk)$ lower bound, but we do not know how to do it, and it might well be impossible. We leave this problem as an intriguing open question.
Recently Braverman et al.~\cite{BEOPV13} confirmed our conjecture that the $k$-player disjointness problem has a lower bound $\Omega(nk)$ in the coordinator model, using a very different method via information complexity.

The second limitation is that symmetrization seems to require proving distributional lower bounds for $2$-player problems, over somewhat-convoluted distributions. This presents some difficulty for the researcher, who needs to start proving lower bounds from scratch and cannot use the literature, since lower bounds in the literature are proved for other distributions. Yao's minimax principle cannot be used, since it only guarantees that there is \emph{some} hard distribution, but it does not guarantee anything for the distribution of interest. This is often only a methodical difficulty, since it is often easy to get convinced that the distribution of interest is indeed hard, but a proof of this must still be found, which is a tedious task. It would be useful if there is some way to circumvent this difficulty, for example by finding a way that standard randomized lower bounds can be used.

The third limitation is that in order to use symmetrization, one needs to find a hard distribution for the $k$-player problem which is symmetric. This is usually impossible when the problem itself is not symmetric, i.e.\ when the players have different roles. For example, one could envision a problem where some of the players get as input elements of some group and the rest of the players get as input integers. However, note that for such problems, symmetrization can still be useful in a somewhat-generalized version. For example, suppose there are two sets of players: in set $P$, the players get group elements, and in set $P'$, the players get integers. Assume each of the sets contains exactly $k/2$ players. To use symmetrization, we would try to find a hard distribution that is symmetric inside of $P$ and also inside of $P'$; namely, a distribution where permuting the players inside $P$ has no effect on the distribution, and similarly for permuting the players inside $P'$. Then, to use symmetrization we can have Alice simulate two random players, $p_i$ and $p_j$, where $p_i$ is from $P$ and $p_j$ is from $P'$; Bob will simulate all the rest of the players. Now symmetrization can be applied. If, alternatively, the set $P$ contained just $3$ players and $P'$ contained $k-3$ players, we can have Alice simulate one of the players in $P'$, Bob can simulate the rest of the players in $P'$, and either Alice or Bob can play the three players from $P$. As can be seen, with a suitable choice of distribution, it should still be possible to apply symmetrization to problems that exhibit some amount of symmetry.

The main topic for future work seems to be to find more setting and problems where symmetrization can prove useful. Recently this technique has found applications in several other statistical, numerical linear algebra and graph problems in the message-passing/coordinator model~\cite{WZ12,WZ13,HRVZ13,WZ14,LSWW14}. We believe it has the potential to be a widely useful tool.

\subsection*{Acknowledgements}
We would like to thank Andrew McGregor for helpful discussions, and for suggesting to study the connectivity problem. We also thank David Woodruff for some insightful comments.

\bibliographystyle{plain}
\bibliography{paper}

\appendix
\section{Omitted Proof for 2-\BITS}
\label{sec:2-BITS-proof}

\noindent\textbf{Lemma \ref{lem:2-BITS} (restated).}  \emph{$\ED_{\zeta_\rho}^{1/3}(\textrm{2-\BITS}) = \Omega(n \rho \log(1/\rho))$.
}

\begin{proof}
Here we will make use of several simple tools from information theory.
Given a random variable $X$ drawn from a distribution $\mu$, we can measure the amount of randomness in $X$ by its entropy $H(X) = - \sum_{x} \mu(x) \log_2 \mu(x)$.  The conditional entropy $H(X \mid Y) = H(X Y) - H(Y)$ describes the amount of entropy in $X$, given that $Y$ exists. The mutual information $I(X : Y) = H(X) + H(Y) - H(XY)$
measures the randomness in both random variables $X$ and $Y$.

Let $\cal P$ be any valid communication protocol. Let $X$ be Alice's (random) input vector. Let $Y$ be Bob's output as the Alice's vector he learns after the communication. Let $\Pi$ be the transcript of $\cal P$. Let $\eps = 1/3$ be the error bound allows by Bob.
First, since after running $\cal P$, with probability at least $1- \eps$, we have $Y = X$, thus Bob knows $X$. Therefore
\[
I(X:Y\ |\ \Pi) \le \eps H(X).
\]
Consequently,
\begin{eqnarray*}
%\hspace{-1.5cm} \mbox{Consequently,} \quad \quad
\eps H(X) &\ge& I(X:Y\ |\ \Pi) \\
&=& H(X\ |\ \Pi) +  H(Y\ |\ \Pi) - H(XY\ |\ \Pi)\\
&=& H(X\Pi) - H(\Pi) + H(Y\Pi) - H(\Pi) \\ & &- (H(XY\Pi) - H(\Pi)) \\
&=& H(X\Pi) + H(Y\Pi) - H(XY\Pi) - H(\Pi) \\
&=& I(X\Pi:Y\Pi) - H(\Pi) \\
&\ge& (1 - \eps) H(X\Pi) - H(\Pi) \\
&\ge& (1 - \eps) H(X) - H(\Pi)
\end{eqnarray*}
Therefore, $\E[\abs{\Pi}] \ge H(\Pi) \ge (1 - 2\eps) H(X) \ge \Omega(n
H(\rho)) \ge \Omega(n\rho \log (1/\rho))$.  
\end{proof}

%%%%%%%%%%%%%%%%%%%%%%%%%%%%%%
\section{Omitted Proofs for the Biased $2$-party Set Disjointness}
\label{sec:2-DISJ-proof}
\noindent\textbf{Lemma \ref{lem:2-DISJ} (restated).}  \emph{When $\mu$ has $|x \cap y| = 1$ with probability $1/t$ then $\ED^{1/100t}_{\mu}(\mbox{2-\DISJ}) = \Omega(n)$.
}

The proof is based on \cite{Raz90}.  Before giving the proof, we first
introduce some notations and a key technical lemma. Define $$A =
\{(x,y)\ :\ (\mu(x,y) > 0) \wedge (x \cap y = \emptyset)\}$$ and
$$B = \{(x,y)\ :\ (\mu(x,y) > 0) \wedge (x \cap y \neq \emptyset)\}.$$ Thus
$\mu(A) = 1 - 1/t$ and $\mu(B) = 1/t$. We need the following key lemma, which is an easy extension of the main lemma in Razbarov~\cite{Raz90} by rescaling the measures on the YES and NO instances.

%The proof will be given shortly.
\begin{lemma}
\label{lem:mono-rect}{\cite{Raz90}}
Let $A, B, \mu$ be defined as above. Let $R = C \times D$ be any rectangle
in the communication protocol. Then we have
$\mu(B \cap R) \ge 1/40t \cdot \mu(A \cap R) - 2^{-0.01n}.$
\end{lemma}

%\begin{proof}
%See Appendix~\ref{sec:2-DISJ-proof}.
%\end{proof}

%\qin{start from here, can be put into the appendix}

%\qin{Now back to main text}

\begin{proof}{(for Lemma~\ref{lem:2-DISJ})}
Let ${\cal R} = \{R_1, \ldots, R_t\}$ be the minimal set of disjoint rectangles in
which the protocol outputs ``$1$", i.e, $x \cap y = \emptyset$. Imagine
that we have a binary decision tree built on top of these rectangles. If we
can show that there exists ${\cal O} \subseteq \cal R$ such that
$\mu(\bigcup_{R_i \in \cal O} R_i) \ge 0.5 \cdot \mu(\bigcup_{R_i \in \cal
  R} R_i)$ and each of $R_i \in {\cal O}$ lies on a depth at least $0.005n$
in the binary decision tree, then we are done. Since $\mu(\bigcup_{R_i \in
  \cal O} R_i) \ge 0.5 \cdot \mu(\bigcup_{R_i \in \cal R} R_i) \ge 0.5
\cdot (\mu(A) - 1/100t) = \Omega(1)$ and querying inputs in each rectangle
in $\cal O$ costs $\Omega(n)$ bits.

We prove this by contradiction. Suppose that there exists ${\cal O'}
\subseteq {\cal R}$ such that $\mu(\bigcup_{R_i \in \cal O'} R_i) > 0.5
\cdot \mu(\bigcup_{R_i \in \cal R} R_i)$ and each of $R_i \in {\cal O'}$
lies on a depth less than $0.005n$ in the binary decision tree. We have the
following two facts.
\begin{enumerate}
\item There are at most $2^{0.005n}$ disjoint rectangles that lie on depths
  less than $0.005n$, i.e., $\abs{\cal O'} \le 2^{0.005n}$.
\item $\mu\left(\bigcup_{R_i \in {\cal O'}}(R_i \cap A)\right) > 0.5 -
  1/100t$.
\end{enumerate}
Combining the two facts with Lemma~\ref{lem:mono-rect} we reach the following contradiction of our error bound.
\begin{align*}
\mu\left(\bigcup_{i=1}^t(R_i \cap B)\right)
\ge& \mu\left(\bigcup_{R_i \in {\cal O'}}(R_i \cap B)\right) \\
\ge& \sum_{R_i \in {\cal
      O'}}(1/40t \cdot \mu(R_i \cap A) - 2^{-0.01n}) \\
>& 1/40t \cdot (0.5 - 1/100t) - 2^{0.005n} \cdot 2^{-0.01n} \\
>& 1/100t \ . \;  
\end{align*}
\end{proof}

%%%%%%%%%%%%%%%%%%%%%%%%%%%%%%%%%%%%%%%%%%%%%%%%%%%%%%%%%%%
%%%%%%%%%%%%%%%%%%%%%%%%%%%%%%%%%%%%%%%%%%%%%%%%%%%%%%%%%%%
\section{Omitted Proofs Graph Connectivity}
\label{sec:CONN-proof-appendix}

We provide here a full proof for the probability of the event $\xi_1$, that both subset of the graph $L$ and $R$ are connected.
\smallskip

\noindent\textbf{Lemma \ref{lem:connected} (restated).}
\emph{$\xi_1$ happens with probability at least $1-1/2n$ when $k \geq 68 \ln n + 1$.}

\begin{proof}
First note that by our construction, both $\abs{L}$ and $\abs{R}$ are $\Omega(n)$ with high probability. To locally simplify notation, we consider a graph $(V,E)$ of $n$ nodes
where edges are drawn in $(k - 1) \ge 68\ln n$ rounds, and each round $n/4$ disjoint
edges are added to the graph.  If $(V,E)$ is connected with
probability $(1-1/4n)$, then by union bound over $\cup_{j=2}^k I_j
\bigcap E_L$ and $\cup_{j=2}^k I_j \bigcap E_R$, $\xi_1$ is true with
probability $(1-1/2n)$. The proof follows four steps.

%The proof follows four steps.  \begin{itemize} \denselist \item[(S1)]
%We show that all points have degree at least $8\ln n$ with
%probability at least $1 - 1/12n$; this uses the first $28 \ln n$
%rounds.  \item[(S2)] We show (conditioned on (S1)) that any subset $S
%\subset V$ of $h < n/10$ points is connected to at least $\min\{h \ln
%n, n/10\}$ distinct points in $V \setminus S$, with probability at
%least $1-1/n^2$.  \item[(S3)] We can iterate (S2) $\ln n$ times to
%show that there must be a single connected component $S_G$ of size at
%least $n/10$, with probability at least $1-1/12n$.  \item[(S4)] We
%can show (conditioned on (S3) and using the last $40 \ln n$ rounds)
%that all points are connected to $S_G$ with probability at least
%$1-1/12n$.  \end{itemize}

\begin{enumerate}
\item[(S1):]  \emph{All points have degree at least $8\ln n$.}
Since for each of $k$ rounds each point's degree increases by $1$ with probability $1/2$, then the expected degree of each point is $14\log n$ after the first $28 \ln n$ rounds.  A Chernoff-Hoeffding bound says that the probability that a point has degree less than $8\ln n$ is at most $2 \exp(-2(6 \ln n)^2/$ $(14 \ln n)) \leq 2 \exp(-5 \ln n) \leq 2/n^5 \leq 1/12n^2$.  Then by the union bound, this holds for none of the $n$ points with probability at least $1-1/12n$.

\item[(S2):] \emph{Conditioned on (S1), any subset $S \subset V$ of $h < n/10$ points is connected to at least $\min\{h \ln n, n/10\}$ distinct points in $V \setminus S$.}
At least $9n/10$ points are outside of $S$, so each point in $S$ expects to be connected at least $(9/10) 8 \ln n \geq 7 \ln n$ times to a point outside of $S$.  Each of these edges occur in different rounds, so they are independent.  Thus we can apply a Chernoff-Hoeffding bound to say the probability that the number of edges outside of $S$ for any point is less than $3 \ln n$ is at most $2 \exp(-2 (4 \ln n)^2/(8 \ln n)) = 2 \exp(-4 \ln n) = 2/n^4$.  Thus the probability that no point in $S$ has fewer than $\ln n$ edges outside $S$ is (since $h < n/10$) at most $1/5n^3$.

If the $h \cdot 3 \ln n$ edges outside of $S$ (for all $h$ points) are drawn independently at random, then we need to bound the probability that these go to more than $n/10$ distinct points or $h \ln n$ distinct points.  Since the edges are drawn to favor going to distinct points in each round, it is sufficient to analyze the case where all of the edges are independent, which can only increase the chance they collide.
In either case $h \ln n < n/10$ or $h \ln n > n/10$ each time an edge is chosen (until $n/10$ vertices have been reached, in which case we can stop), $9/10$ of all possible vertices are outside the set of edges already connected to.  So if we select the $3 h \ln n$ edges one at a time, each event connects to a distinct points with probability at least $9/10$, so we expect at least $(9/10) (3 h \ln n) > 2 h \ln n$ distinct points.  Again by a Chernoff-Hoeffding bound, the probability that fewer than $h \ln n$ distinct points have been reached is at most $2 \exp(-2 (h \ln n)^2/(3 h n)) \leq 2 \exp(-(2/3) h \ln n) < 2 \cdot$ $\exp(-5 \ln n) \leq 2/n^5 \leq 1/5n^2$ (for $h \geq 8$).
Together the probability of these events not happening is at most $1/2n^2$.

\item[(S3):]  \emph{There is a single connected component $S_G$ of size at least $n/10$.}
Start with any single point, we know from Step 1, its degree is at least $8 \ln n$.  Then we can consider the set $S$ formed by these $h_1 = 8 \ln n$ points, and apply Step 2 to find another $h_1 \ln n = 8 \ln^2 n = h_2$ points; add these points to $S$.  The process iterates and at each round $h_i = 8 \ln^i n$, by growing only from $h_i$ the newly added points.  So, by round $i = \ln n$ the set $S = S_G$ has grown to at least size $n/10$.  Taking the union bound over these $\ln n$ rounds shows that this process fails with probability at most $1/12n$.
% (for $n \geq 17$).

\item[(S4):] \emph{All points in $V \setminus S_G$ are connected to $S_G$.}
Each round each point $p$ is connected to $S_G$ with probability at least $1/20$.  So by coupon collector's bound, using the last $40 \ln n$ rounds all points are connected after $2 \ln n$ sets of $20$ rounds with probability at least $1-1/12n$.
\end{enumerate}
By union bound, the probability of steps (S1), (S3), and (S4) are successful is at least $1-1/4n$, proving our claim. 
\end{proof}

\end{document}